\documentclass[journal]{IEEEtran}
\usepackage{cite}
\usepackage{amsmath,amssymb,amsfonts,amsthm}
\usepackage{algorithm}
\usepackage{algpseudocode}
\usepackage{graphicx}
\usepackage{tikz}
\usepackage{graphicx}
\usepackage{pgfplots}
\usepackage{mathrsfs}  
\usepackage{caption}
\usepackage{subcaption}
\usepackage{comment}
\usepackage{float}
\usepackage{physics}
\usetikzlibrary{spy,backgrounds}
\usepackage{lipsum}
\usepackage{siunitx}
\newtheorem{theorem}{Theorem}
\newtheorem{lemma}{Lemma}
\newtheorem{corollary}{Corollary}

\newtheorem{definition}{Definition}

\usepackage{tabularx}
\usepackage{booktabs}  % For better lines
\usepackage{caption}
\theoremstyle{remark}
\newtheorem{remark}{Remark}

\usepackage{tikz}
\usepackage{caption}
\usepackage{float}
\usetikzlibrary{decorations.pathreplacing}

\newcommand\numeq[1]%
  {\stackrel{\scriptscriptstyle(\mkern-1.5mu#1\mkern-1.5mu)}{=}}
  \newcommand\numl[1]%
  {\stackrel{\scriptscriptstyle(\mkern-1.5mu#1\mkern-1.5mu)}{<}}
\newcommand\numleq[1]%
  {\stackrel{\scriptscriptstyle(\mkern-1.5mu#1\mkern-1.5mu)}{\leq}}
\newcommand\numgeq[1]%
  {\stackrel{\scriptscriptstyle(\mkern-1.5mu#1\mkern-1.5mu)}{\geq}}
\usepackage{mathtools}
\usepackage{breqn}

\usepackage{xcolor}
\def\BibTeX{{\rm B\kern-.05em{\sc i\kern-.025em b}\kern-.08em
    T\kern-.1667em\lower.7ex\hbox{E}\kern-.125emX}}
    
    \usepackage{cuted}
\usepackage{lipsum}

\ifCLASSINFOpdf 
\else  
\fi

\usepackage{times}
\usepackage{tikz,graphicx}
\usepackage{amssymb}
\usepackage{amsthm}
\usepackage{algpseudocode}
\usepackage{amsxtra}
\usepackage{algorithm}
\usepackage{amsmath}
\usepackage{hyperref}
\usepackage{verbatim}
\usepackage{dsfont}
\usepackage{epsfig}
\usepackage{multirow}
\usepackage{amsmath}
\usepackage{amssymb}
\usepackage{mathtools}
\usepackage{cite}

\def\u{{\mathbf u}}

\def\y{{\mathbf y}}

\def\0{{\mathbf 0}}
\def\1{{\mathbf 1}}

\def\beq{\begin{equation}}
	\def\eeq{\end{equation}}
\def\beqa{\begin{eqnarray*}}
	\def\eeqa{\end{eqnarray*}}

\usepackage{xcolor}

\usepackage{xcolor}

\def\u{{\mathbf u}}

\def\y{{\mathbf y}}

\def\0{{\mathbf 0}}
\def\1{{\mathbf 1}}

\def\beq{\begin{equation}}
	\def\eeq{\end{equation}}
\def\beqa{\begin{eqnarray*}}
	\def\eeqa{\end{eqnarray*}}

\author{Jaswanthi Mandalapu, Deepak Charan, V.~Arvind Rameshwar, and Nir Weinberger}
\title{Fundamental Limits of Sequence Reconstruction Problems in Immunogenomics}
\begin{document}

\maketitle
\begin{abstract}
The goal of personalized immunogenomics is to recover an individual's germline immunoglobulin gene segments from `\emph{repertoire sequences}' altered by trimming, extension, and mutation. In this work, we study the fundamental trace complexity (or the number of samples/traces required for accurate reconstruction) of D-gene reconstruction under three biologically motivated trace-generation models introduced by Bhardwaj et al. (2021) and develop practical algorithms for reconstruction---problems left open in the original work. First, for the TrimSuffixAndExtend model, we establish the optimal trace complexity to be $\Theta(n)$, and develop a low-complexity Prefix-Filtered Mode (PFM) decoder that achieves this scaling. Second, for the closely related two-sided TrimAndExtend model, we show that the optimal trace complexity is instead $\Theta(n^2)$, and is achieved by the simple Bit-Wise Mode (BWM) decoder. Third, for the SuffixExtend$_t$(TrimSuffix) model, we establish polynomially separated lower and upper bounds on trace complexity. Our results follow from information-theoretic lower bounds coupled with tight analyses of the proposed reconstruction algorithms. 
%Broadly, our results demonstrate how the biological corruption mechanism can govern the number of observations required for reliable sequence reconstruction.
% The goal of personalized immunogenomics is to recover an individual's germline immunoglobulin gene segments from ``repertoire sequences" altered by trimming, extension, and mutation. In this work, we study the sample complexity of this task, i.e., the number of repertoire sequences required to reconstruct the germline genes with high probability, by working with structured trace-generation channels recently introduced in \cite{bhardwaj}. For the TrimSuffixAndExtend channel considered in this work, we establish that the trace complexity scales linearly in the gene segment lengths and provide a low-complexity decoder---the Prefix-Filtered Mode (PFM) decoder---that asymptotically achieves this complexity. For the two-sided TrimAndExtend channel of \cite{bhardwaj}, we show that, interestingly, the trace complexity scales quadratically in the gene segment lengths (in sharp contrast to TrimSuffixAndExtend), and is achieved by the simple, Bit-Wise Mode (BWM) decoder. The proofs follow via information-theoretic lower bounds and refined analyses of the decoders, via concentration arguments. These results establish that biologically structured corruption in repertoire sequences can be recovered from, optimally, using low-complexity, practical algorithms.
\end{abstract}

\section{Introduction}
\label{sec:introduction}

Recovering an unknown discrete sequence from multiple noisy or corrupted observations is a statistical inference problem that arises in a variety of applications. Of particular interest is the question of {how many such observations are required to reconstruct the underlying sequence with a prescribed probability of error}. In this work, we study this question for a class of trimming and extension mechanisms motivated by germline gene reconstruction in immunogenomics, where multiple repertoire sequences provide differently corrupted observations of an underlying gene sequence.

 Informally speaking, the human immune system generates a large diversity of antibody sequences by assembling and modifying a collection of underlying gene segments. Modern sequencing technologies provide access to large collections of these antibody sequences, but not to the gene segments that generated them. Indeed, the underlying biological process can remove nucleotides (the basic building blocks of DNA) from these segments and insert new ones at their boundaries. Thus, what is observed is not the underlying gene sequence itself, but rather a collection of its (potentially differently) corrupted realizations.

 More concretely, each antibody sequence is assembled from variable (V), diversity (D), and joining (J) germline gene segments through a biological process known as V(D)J recombination. High-throughput adaptive immune receptor repertoire sequencing (AIRR-seq) provides access to large collections of the resulting antibody sequences from an individual. Reconstructing the germline segments underlying these sequences is a key inference problem in repertoire analysis \cite{yaari-kleinstein2015,marcou2018igor}. The problem is rendered particularly challenging by incomplete population reference databases, the highly polymorphic nature of immunoglobulin loci, and the fact that an individual typically carries only a subset of the catalogued alleles \cite{gadala-maria2015tigger,corcoran2016igdiscover,ralph-matsen2019germline}.

 Within this broader setting, the reconstruction of D segments presents an intriguing challenge. In particular, D segments are inherently short, and are trimmed at both ends during V(D)J recombination and subsequently surrounded by non-templated nucleotides. Consequently, the underlying D-gene sequence is not observed directly; instead, one obtains several partial and differently corrupted versions of it. This has motivated specialized computational approaches such as IgScout and MINING-D \cite{safonova-pevzner2019igscout,bhardwaj2020miningd}. More recently, \cite{bhardwaj} introduced an information-theoretic formulation that captures the underlying biological operations through random trimming and extension mechanisms, with each resulting corrupted observation referred to as a \emph{trace}. This formulation brings to the fore a natural question: how many independently generated traces are necessary and sufficient to reliably reconstruct the underlying sequence with a given accuracy? While \cite{bhardwaj} introduced the relevant trace-generation models, a characterization of their fundamental trace complexities, together with low-complexity reconstruction algorithms that attain them, remained largely open. It is precisely these questions that we take up in the present work.

 In information-theoretic terms, the underlying gene segment may be viewed as the input to a \emph{channel}, and each independently generated trace as an output of that channel; we adopt this terminology in what follows. The fundamental quantity of interest is then the \emph{trace complexity} (or sample complexity), namely, the minimum number of independently generated traces required to reconstruct the underlying sequence with a prescribed probability of error.

\subsection{Relevant Literature} The problem of reconstructing a sequence from multiple corrupted observations has a rich history. Information-theoretic lower bounds and recovery algorithms with matching sample complexities have been studied in the machine-learning literature in the context of sparse estimation and latent-structure learning \cite{aksoylar-saligrama2014aistats,ghoshal-honorio2017aistats}, including exact structure recovery when the observations themselves are corrupted \cite{nikolakakis2019aistats}. Closely related is the classical problem of trace reconstruction over a \textit{deletion channel}. Beginning with Levenshtein's combinatorial reconstruction problems \cite{levenshtein,levenshtein-2}, a substantial body of work has investigated recovery from random subsequences generated by an i.i.d.\ deletion channel \cite{batu,haeupler,de-odonnell-servedio2017,peres-zhai,coded-trace,chase-2}. A parallel line of work in information theory concerns itself with optimally combining multiple traces generated by memoryless channels and characterizing the resulting achievable information rates \cite{mitz,land-huber,extremes,var-nir}. Furthermore, similar reconstruction questions also arise naturally in DNA-based data storage, where information is encoded in collections of DNA strands and subsequently recovered from potentially corrupted copies. Several experimental works have demonstrated practical DNA-storage systems \cite{dnachurch,dnagoldman,dnaerlich,dnaorganick}. Complementing these developments, a growing information-theoretic literature has sought to characterize the fundamental limits of storage and reconstruction in such systems \cite{shomorony,nir_merhav,sima,mao,nanopore-1,ar-paper-5}.

Unlike much of the existing literature, a salient feature of the models considered here is the statistical dependence induced by the underlying trimming mechanisms. A single trimming event affects a contiguous portion of the sequence, rendering the resulting coordinate errors neither independent nor spatially homogeneous. More importantly, the structure of this dependence can itself be exploited for reconstruction: information gleaned from one portion of a trace can reveal the reliability of its subsequent coordinates. As we shall see, this seemingly simple observation has rather striking consequences for trace complexity, and leads to markedly different scaling laws across the models considered. 
\subsection{Our Contributions}
In summary, our main contributions are as follows.
\begin{itemize}
\item We derive information-theoretic lower bounds on the trace complexities of the TrimSuffixAndExtend, TrimAndExtend, and SuffixExtend$_t$(TrimSuffix) models introduced by \cite{bhardwaj}.

\item For TrimSuffixAndExtend ($W_1$), we establish that the trace complexity $T_\delta(n)$ is $\Theta(n)$. While the coordinate-wise Bit-Wise Mode (BWM) decoder requires $O(n^2\log n)$ traces, our Prefix-Filtered Mode (PFM) decoder exploits the nested prefix structure and achieves the optimal $O(n)$ scaling.

\item For TrimAndExtend ($W_2$), we establish that the trace complexity $T_\delta(n)$ is $\Theta(n^2)$, with the simple BWM decoder itself achieving the optimal scaling. Thus, despite their closely related corruption mechanisms, $W_1$ and $W_2$ exhibit fundamentally different trace-complexity scalings.

\item For SuffixExtend$_t$(TrimSuffix) ($W_3$), we establish an $\Omega(n)$ lower bound and an $O(n^3)$ algorithmic upper bound, leaving its precise trace complexity open.

\end{itemize}

The remainder of the paper is organized as follows. Sec.~2 introduces the channel models and performance criteria. Secs.~3 and~4 establish matching lower and upper bounds for $W_1$ and $W_2$, respectively. Sec.~5 derives lower and upper bounds for $W_3$. All proofs are deferred to the appendices.

\section{Channel Models and Preliminaries}

In this section, we review the channel models introduced in \cite[Sec. III]{bhardwaj} and establish the notation used throughout the paper. Our focus is on the reconstruction of \textsc{D} genes in the immunoglobulin locus, which forms the basic building block of the more general \textsc{V}-\textsc{D}-\textsc{J} reconstruction framework proposed in \cite[Sec. III]{bhardwaj}.

% While [1] introduced biologically motivated trace-generation models for this problem, their fundamental trace complexity remains largely unknown. In this paper, we derive lower bounds on the trace complexity of these channels and develop reconstruction algorithms that achieve the derived bounds.

% In this section, we recapitulate the channel models presented in \cite[Sec. III]{bhardwaj}. We focus exclusively on channel models for the task of reconstructing the D gene in the immunoglobin locus; more complicated models for reconstructing the V, D, and J genes in the immunoglobin locus, based on concatenations of outputs of channel models for the D gene can also be found in \cite[Sec. III]{bhardwaj}.

Let $n$ denote the length of the \textsc{D} gene. The gene is modeled as a sequence $\mathbf{x} \in \mathcal{X}^n$, where $\mathcal{X} = \{0,1,\ldots, q-1\}$ is an alphabet of some fixed size $q>1$. Following this, we define our channel models below:
%We define the memoryless channel SC$(\epsilon)$, for a fixed $\epsilon\in (0,1)$, to have input and output alphabets equalling $\mathcal{X}$ and a channel law $P(x|x) = 1-\epsilon$, and $P(y|x) = \frac{1-\epsilon}{q-1}$, for all $y\neq x$. 
%We are interested in the following channel models:
\begin{enumerate}
	\item \emph{TrimSuffixAndExtend (Channel $W_1$)}: In this channel model, first an integer ${R\sim \text{Unif}([0:n])}$ is sampled, and the last $R$ symbols of $\mathbf{x}$ are replaced independently by symbols drawn uniformly from $\mathcal{X}$.
	\item \emph{{TrimAndExtend} (Channel $W_2$)}: Here, a pair of non-negative integers $(R_1,R_2)$ is sampled uniformly from the set of all pairs satisfying ${R_1 + R_2 \leq n}$. Then, the first $R_1$ symbols and the last $R_2$ symbols of $\mathbf{x}$ are replaced independently by symbols drawn uniformly from $\mathcal{X}$.
	\item \emph{SuffixExtend$_{t}$(TrimSuffix) (Channel $W_3$)}: For this channel, an integer $R\sim \text{Unif}([0:n])$ is sampled, and the last $R$ symbols of $\mathbf{x}$ are deleted, yielding a trimmed sequence $\mathbf{x}'$. Next, for a fixed integer $t\geq 1$, an integer $E\sim \text{Unif}([0:t])$ is sampled and $E$ symbols, drawn independently and uniformly from $\mathcal{X}$, are appended at the end of $\mathbf{x}'$.
%	\item Mutate$_\epsilon$(TrimSuffixAndExtend) (or channel $W_4$): The sequence $\mathbf{x}$ is first passed through TrimSuffixAndExtend, to yield the sequence $\mathbf{x}'$. Each symbol of $\mathbf{x}'$ is then passed independently through an SC($\epsilon$).
\end{enumerate}

Note that while \cite{bhardwaj} introduced these biologically motivated trace-generation models, their fundamental trace complexity remains largely unexplored. In this paper, we derive lower bounds on the trace complexity of the channels described above and develop reconstruction algorithms that, attain the derived bounds  for channels $W_1$ and $W_2$, and provide upper bounds on the trace complexity for channel $W_3$. Our analysis provides a first step towards understanding the fundamental limits of a broader class of immunogenomic trace-generation channels.
We believe many techniques developed here can also be carried out for generalized channel models that additionally allow symbol \emph{mutations} or \emph{substitutions}, following the corruption mechanisms described above.

% Next, given a channel $W$, we define
% \[
% d_n(W):=\min_{\mathbf{u}\neq \mathbf{u}'\in \mathcal{X}^n} d_{\text{TV}}\left(W(\cdot|\mathbf{u}), W(\cdot|\mathbf{u}')\right)
% \]
% to be the smallest total variational distance between channel transition probabilities corresponding to distinct channel input sequences.
We now proceed to introduce the notation and performance metrics used throughout the paper.

Throughout the paper, $\mathcal{X}$ denotes the input alphabet of size $q$, and $\mathcal{X}^n$ denotes the set of all length-$n$ sequences over $\mathcal{X}$. We use $\mathcal{X}^*=\bigcup_{m\ge0}\mathcal{X}^m$ to denote the set of all finite-length sequences over $\mathcal{X}$. Next, random variables are denoted by uppercase letters (e.g., $R$ and $E$), while their realizations are denoted by the corresponding lowercase letters. We use boldface symbols, such as $\mathbf{0}$, to denote constant vectors; in particular, $\mathbf{0}$ denotes the all-zeros sequence of length $n$. Further, the indicator function is denoted by $\mathbf{1}\{\mathcal{E}\}$, which equals one if the event $\mathcal{E}$ occurs and zero otherwise. For a channel $W$, the conditional probability of observing an output trace $y\in\mathcal{X}^*$ given the transmitted sequence $x\in\mathcal{X}^n$ is denoted by $W(y|x)$. The notation $\mathrm{Unif}(\mathcal{S})$ denotes the uniform distribution over a finite set $\mathcal{S}$, and all logarithms are taken to base two unless stated otherwise.
\begin{definition}
	A sequence reconstruction algorithm, $\mathcal{A}$, takes as input traces $\mathbf{y}^{(1)},\ldots,\mathbf{y}^{(N)}\in \mathcal{X}^\star$, for some $N\geq 1$, each obtained independently from $\mathbf{x}$ via a channel $W$, and returns as output a sequence (gene) $\widehat{\mathbf{x}}\in \mathcal{X}^n$.
\end{definition}

We next introduce the notions of average-case error probability and trace complexity, beginning with the definition of the (optimal) maximum aposteriori probability (MAP) decoder.

\begin{definition} (Maximum aposteriori probability (MAP) decoder) Given N traces $(\mathbf{y}^{(1)},\ldots,\mathbf{y}^{(N)})$, the MAP decoder outputs
$$
\hat{\mathbf{x}}^*_{\mathrm{MAP}}
\in
\arg\max_{\mathbf{x}\in\mathcal{X}^n}
\prod_{i=1}^{N} W(\mathbf{y}^{(i)}|\mathbf{x}),
$$
where the maximization is over all possible input sequences $\mathbf{x} \in\mathcal{X}^n$. The above expression follows from Bayes' rule under the assumption of a uniform prior on $\mathcal{X}^n$.

\end{definition}

%\begin{definition}[Worst-case error probability and trace complexity]
%	The worst-case error probability of a sequence reconstruction algorithm $A$ over the channel $W$, given $N$ traces, is 
%	\[
%	P_{e,A}^{(N)}:= \max_{\mathbf{x}\in \mathcal{X}^n} P_{e,A}^{(N)}(\mathbf{x}),
%	\]
%	where
%	\[
%	P_e^{(N)}(\mathbf{x}):= \sum_{(\mathbf{y}_1,\ldots,\mathbf{y}_N)\in \mathcal{X}^\star} \prod_{i=1}^{N}W(\mathbf{y}_i|\mathbf{x})\cdot \mathds{1}\{A(\mathbf{y}_1,\ldots,\mathbf{y}_N) \neq \mathbf{x}\}.
%	\]
%	The worst-case trace complexity, for a fixed error probability $\delta\in (0,1)$ is given by
%	\[
%	T_\delta(n) = \min\{N: P_{e,\textsc{MAP}}^{(N)}\leq \delta\}.
%	\]
%\end{definition}

\begin{definition}
	The average-case error probability of a sequence reconstruction algorithm, $\mathcal{A}$, over the channel $W$, given $N$ traces is defined as
	$
{p}_{e,\mathcal{A}}^{(N)}:= \frac{1}{q^n} \sum_{\mathbf{x}\in \mathcal{X}^n} p_{e,\mathcal{A}}^{(N)}(\mathbf{x}),
	$
	where $p_{e,\mathcal{A}}^{(N)}(\mathbf{x})$ is defined as
\begin{align*}
\sum_{(\mathbf{y}^{(1)},\ldots,\mathbf{y}^{(N)})\in \mathcal{X}^\star} \prod_{i=1}^{N}W(\mathbf{y}^{(i)}|\mathbf{x}) \mathds{1}\{\mathcal{A}(\mathbf{y}^{(1)},\ldots,\mathbf{y}^{(N)}) \neq \mathbf{x}\}.
\end{align*}
\end{definition}
 \begin{definition}	For a target error probability $\delta\in(0,1)$, the trace complexity of a sequence reconstruction algorithm $\mathcal{A}$ is defined as
$$
{T}_{\delta, \mathcal{A}}(n) = \min\{N: {p}_{e,\mathcal{A}}^{(N)}\leq \delta\}.
$$
In particular, letting MAP denote the maximum a posteriori decoder defined above, we define the trace complexity of the channel as 
$$
    T_{\delta}(n) := T_{\delta,\mathrm{MAP}}(n) = \min\{ N: p_{e, \mathrm{MAP}}^{(N)} \leq \delta \}.
$$
Further, since the MAP decoder minimizes the average error probability, we have
$
    T_{\delta}(n) \leq T_{\delta, \mathcal{A}}(n),
$
for any sequence reconstruction algorithm $\mathcal{A}.$
\end{definition}
\begin{definition}(Trace complexity conditioned on an input sequence)
	For a fixed sequence $\mathbf{x}\in\mathcal X^n$, the trace complexity conditioned on $\mathbf{x}$ is defined as
	$
	T_\delta(n;\mathbf{x}):= \min\{N: p_{e,\textsc{MAP}}^{(N)}(\mathbf{x})\leq \delta\}.
	$
\end{definition}
Further, the following lemma shows that, for the channels under consideration, the trace complexity is independent of the transmitted sequence.
\begin{lemma}
	\label{lem:worst-av}
	For the channels $W_1$ through $W_3$, $T_\delta(n) = {T}_\delta(n;\mathbf{x})$, for all $\delta\in (0,1)$, $n\geq 1$, and $\mathbf{x}\in \mathcal{X}^n$.
\end{lemma}
\begin{proof}
	For each of the channels $W_1$ through $W_3,$ the distribution of the corruption pattern is independent of the transmitted sequence and is symmetric over the input alphabet. Consequently, the conditional error probability $p_{e,\mathrm{MAP}}^{(N)}(\mathbf{x})$ is identical for all $\mathbf{x} \in \mathcal{X}^n$, completing the proof.
\end{proof}
%Following Lemma \ref{lem:worst-av} and its proof, we arrive at the fact that $T_\delta(n) = \overline{T}_\delta(n) = T_\delta(n;\mathbf{0})$, where 
%\[
%T_\delta(n;\mathbf{0}):= \min\{N: \overline{P}_e^{(N)}(\mathbf{0})\leq \delta\}
%\]
%is the minimum number of traces required to reconstruct the all-zeros input sequence with error probability at most $\delta\in (0,1)$. 
From Lemma~\ref{lem:worst-av}, in what follows, we hence restrict attention to the setting where the all-zeros sequence is transmitted over channels $W_1$ through $W_3$, and obtain bounds on the quantity $T_\delta(n)$ via bounds on $T_\delta(n;\mathbf{0})$. To derive lower bounds on the trace complexity, we use two information-theoretic techniques. For channels $W_1$ and $W_2$, we use Fano's method together with the KL divergence between suitably chosen channel output distributions. For channel $W_3$, we instead use a binary hypothesis testing argument based on the total variation distance between the corresponding channel output distributions. Below, we first introduce the KL-divergence-based lower-bounding principle that will be used in the subsequent analysis of channels $W_1$ and $W_2$.

For a coordinate $j \in [n]$ and a symbol $p \in \mathcal{X},$ let $\mathbf{x}^{(j,p)} = \left(0,\ldots,0,p,0,\ldots,0\right),$ where $p$ appears in the $j-$th coordinate. In particular, $\mathbf{x}^{(j,0)} = \mathbf{0}.$ Next, let $P_{n,j,p}$ denote the probability distribution of a single trace generated by the channel  $W$ when the transmitted sequence is $\mathbf{x}^{(j,p)},$ i.e., 
$
    P_{n,j,p}(\mathbf{y}) = W(\mathbf{y}|\mathbf{x}^{(j,p)}), \mathbf{y} \in \mathcal{X}^*.
$ Further, let $P_{n,j,p}^{(N)} = P_{n,j,p}^{\otimes N}$ denote the distribution of $N$ independent traces generated from $\mathbf{x}^{(j,p)}.$ The following lemma then provides the Fano-based lower-bounding framework that will be used for channels $W_1$ and $W_2.$
\begin{lemma}\label{lemma-lowerbound-channel-w1-w2}
    For any fixed $\delta \in \left(0, 1 - 1/q\right),$
    $$
    T_{\delta}(n) \geq \max_{j \in [n]} \frac{\log{q} - h_b(\delta) - \delta \log{(q-1)}}{\max_{p \in \mathcal{X}\setminus \{0\} }D(P_{n,j,0} || P_{n,j,p})}.
    $$
\end{lemma}
For completeness, we provide a proof of the lemma above in Appendix~\ref{proof-lemma-lower-bound}.
% Note that Lemma~\ref{lemma-lower-bound} provides a general converse framework for deriving lower bounds on the trace complexity of the channel models under consideration. 

Note that Lemma~\ref{lemma-lowerbound-channel-w1-w2} reduces the sequence reconstruction problem to a $q$-ary estimation problem associated with a single input coordinate. For deriving the lower bounds of channels $W_1$ and $W_2$, we use this result by appropriately choosing the coordinate $j$ and evaluating the corresponding KL divergence. For channel $W_3$, we instead work directly with the total variation distance between channel output distributions, as we will see in Sec.~\ref{sec-w3}.

We then develop reconstruction algorithms which, for channels $W_1$ and $W_2$, are such that their trace complexities match the derived information-theoretic lower bounds up to constant factors. Together, these results help establish asymptotically tight characterizations of the trace complexities of $W_1$ and $W_2$; for $W_3$, however, we only obtain lower and upper bounds on the trace complexity, and leave for future work the task of obtaining a tight characterization. Furthermore, our reconstruction algorithms rely only on \emph{simple, low-complexity procedures}, in contrast to MAP decoding, which requires a time complexity that is exponential in $n$.
% We next introduce a simple algorithm, \textsc{FindMode} (see Algorithm \ref{alg:mode}), that returns the trace occurring most often. As we shall see, all our reconstruction algorithms are either \textsc{FindMode} itself, for a certain values of $N$, or are simple variants thereof.
% \begin{algorithm}[t]
% 	\caption{Sequence reconstruction algorithm}
% 	\label{alg:mode}
% 	\begin{algorithmic}[1]	
		%		\State: Construct a one-one mapping $\phi_m: \mathbb{F}_2^{N_m-K_m}\to \mathcal{C}_{\log_2\left(\frac{N_m-K_m}{2^{-\left \lceil \log_2(d+1)\right \rceil}R}\right)}(R)$, where $\mathcal{C}_{\log_2\left(\frac{N_m-K_m}{2^{-\left \lceil \log_2(d+1)\right \rceil}R}\right)}(R)$ is constructed as in \eqref{eq:rmlb1}.
% 		\Procedure{\textsc{FindMode}}{$\mathbf{y}_1,\ldots,\mathbf{y}_{N}$}
% 		\State Return the sequence occurring most often in $(\mathbf{y}_1,\ldots,\mathbf{y}_{N})$.
% 		\EndProcedure	
% 	\end{algorithmic}
% \end{algorithm}

%In the sections that follow, we take up each of the channel models above in turn and 
\section{TrimSuffixAndExtend (Channel $W_1$)}
In this section, we first establish an $\Omega(n)$ lower bound on the trace complexity of channel $W_1$ using Lemma~\ref{lemma-lowerbound-channel-w1-w2}. We then present two low-complexity reconstruction algorithms, namely, Bit-Wise Mode (BWM) and Prefix-Filtered Mode (PFM). To build intuition about the problem, we first consider BWM and show that it achieves a trace complexity of $O(n^2 \log{n})$, providing a simple baseline reconstruction procedure. We then argue that PFM, on the other hand, achieves a trace complexity of $O(n)$, which when combined with the lower bound, establishes the trace complexity to be $\Theta(n)$. 

% Nevertheless, we present BWM first, as it serves as a useful conceptual building block and plays a central role in our analysis of channel $W_2$, where it is shown to attain the corresponding lower bound.

% We now proceed to present the lower bound on trace complexity of $W_1.$
\begin{theorem}\label{thm:lowerbound:w1}
For channel $W_1,$ for any $\delta \in (0,1-1/q),$ we have $T_{\delta}(n) = \Omega(n).$
\end{theorem}
The proof proceeds with the help of Lemma \ref{lemma-lowerbound-channel-w1-w2} and is in Appendix \ref{app:thm-1}.

\begin{algorithm}[t]\label{alg:bwm}
\caption{Bit-Wise Mode (BWM)}
\label{alg:bwm}
\begin{algorithmic}[1]
\Require Traces $y^{(1)},\ldots,y^{(N)}\in\mathcal X^n$
\Ensure Estimate $\widehat{\mathbf{x}}\in\mathcal X^n$

\For{$j=1,\ldots,n$}
    \State
    $\widehat{x}_j \in
    \arg\max_{a\in\mathcal X}
    \sum_{i=1}^{N}
    \mathbf{1}\{y_j^{(i)}=a\}$
\EndFor
\State
\Return $\widehat{\mathbf{x}}=(\widehat{x}_1,\ldots,\widehat{x}_n)$
\end{algorithmic}
\end{algorithm}
Next, we present the Bit-Wise Mode (BWM) algorithm, and derive its trace complexity under channel $W_1.$ BWM reconstructs the input sequence by simply returning, for each position, the symbol that appears most frequently among the corresponding positions of the observed traces; see Algorithm~\ref{alg:bwm} for a formal description. We now characterize its performance over channel $W_1.$

\begin{theorem}
\label{thm:bwm-w1}
For any fixed $\delta\in(0,1)$, over channel $W_1$, the BWM algorithm reconstructs the input sequence with probability at least $1-\delta$ using
$
N = O(n^2\log n)
$
traces. Consequently, $T_{\delta, \mathrm{BWM}}(n) = O(n^2 \log{n}).$
\end{theorem}
We prove Theorem \ref{thm:bwm-w1} by  first characterizing the marginal distribution of each coordinate under channel $W_1$ and show that the true symbol has a strictly larger occurrence probability than any incorrect symbol. We then use Hoeffding's inequality and a union bound to bound the probability that an incorrect symbol appears at least as frequently as the true symbol. The complete proof is in Appendix~\ref{proof:thm-bwm-w1}.

\begin{remark}
Note that the trace complexity of BWM is primarily decided by the last coordinate. Specifically, under channel $W_1$, the symbol at position $j$ is preserved with probability $(n-j+1)/(n+1)$, which decreases as $j$ increases. In particular, the last coordinate is left uncorrupted with probability \emph{only} $1/(n+1)$. Further, since BWM estimates each coordinate independently, its performance is determined by this ``weakest" coordinate, resulting in the $O(n^2\log n)$ trace complexity derived above.
\end{remark}

The above remark indicates that a decoding rule that treats coordinates independent of each other is inherently suboptimal for $W_1$. Indeed, traces contain significantly more information about earlier coordinates than later ones. This motivates the development of a novel reconstruction procedure, which we call Prefix-Filtered Mode (PFM), that exploits the structural relationship among coordinates instead of estimating them independently.

The PFM decoder is provided in Algorithm~\ref{alg:pfm}. At a high level, PFM works as follows: at each step, it treats the current estimate of a prefix as the ``true," uncorrupted prefix, and then identifies traces that agree with this estimate and provide an uncorrupted observation of the next coordinate. In particular, the next coordinate is then estimated by applying a mode operation to this ``filtered" collection of traces, at the coordinate of interest. 
% More importantly, by progressively filtering traces using information obtained from earlier coordinates, PFM avoids the bottleneck encountered by BWM at the later coordinates. 

\begin{algorithm}[t]
\caption{Prefix-Filtered Mode (PFM)}
\label{alg:pfm}
\begin{algorithmic}[1]
\Require Traces $y^{(1)},\ldots,y^{(N)}\in\mathcal X^n$
\Ensure Estimate $\widehat{\mathbf{x}}\in\mathcal X^n$

\State $\widehat{\mathbf{x}}\gets \emptyset$
\State $\mathcal S_0\gets [N]$

\For{$j=1,\ldots,n$}
    \State
    $\widehat{x}_j \in
    \arg\max_{a\in\mathcal X}
    \sum_{i\in\mathcal S_{j-1}}
    \mathbf 1\{y_j^{(i)}=a\}$
    \State
    $\mathcal S_j
    \gets
    \{i\in\mathcal S_{j-1}:y_j^{(i)}=\widehat{x}_j\}$
\EndFor

\State
\Return $\widehat{\mathbf{x}}
=
(\widehat{x}_1,\ldots,\widehat{x}_n)$

\end{algorithmic}
\end{algorithm}

We now show that PFM achieves order-optimal trace complexity for channel $W_1$, by proving that it achieves a trace complexity of $O(n)$.

\begin{theorem}\label{Thm:PFM} 
       For any constant $\delta \in (0,1),$ over channel $W_1$, the PFM algorithm reconstructs the input sequence with probability at least $1-\delta$ using
$
N = O(n)
$ traces. Consequently, $T_{\delta, \mathrm{PFM}}(n) = O(n).$ Furthermore, since $T_{\delta}(n) \leq T_{\delta, \mathrm{PFM}}(n)$, channel $W_1$ has trace complexity $T_{\delta} (n) = O(n).$ 
\end{theorem}

We first present an outline of the proof strategy of the theorem above. Note that by Lemma~\ref{lem:worst-av}, it suffices to consider the all-zero input sequence $\mathbf{x}=\mathbf{0}.$ The proof of Theorem~\ref{Thm:PFM} then relies on the three main intermediate steps.
\begin{itemize}
    \item \emph{Step 1}: First, we characterize the probability that a trace agrees with the true prefix after $k$ coordinates. This provides an estimate of the number of traces available to PFM at Step $k+1$ of Algorithm \ref{alg:pfm}.
    \item \emph{Step 2}: Second, we show that, among those traces that agree with the estimated prefix, the correct symbol at the next coordinate has a constant advantage over every incorrect symbol.
    \item \emph{Step 3}: Third,  we use Hoeffding's inequality to bound the probability of an incorrect mode decision at each step. 
\end{itemize}
Finally, we show that combining these bounds over all coordinates concludes the proof. We now establish each of the above steps in detail.

\emph{Step 1}: 
For each $i \in \{1,\ldots, N\}$ and $k\in\{0,1,\ldots,n-1\}$, define the event
$
A_k^{(i)}
:=
\{y_1^{(i)}=0,\ldots,y_k^{(i)}=0\},
$ with the convention that $A_0^{(i)}$ denotes the entire sample space. Thus, $A_k^{(i)}$ is the event that the $i^{th}$ trace agrees with the true prefix of length $k$. Further, note that, since the traces are independent and identically distributed, $P(A_k^{(i)})$ is the same for all $i\in [N]$; we denote this common probability by $P(A_k)$. 

\begin{lemma}
\label{lem:ak}
Under channel $W_1$, for every $k\in\{0,1,\ldots,n-1\}$, we have
$$
\frac{n-k+1}{n+1}
\le
P(A_k)
\le
\frac{1}{n+1}
\left(
n-k+1+\frac{1}{q-1}
\right).
$$
\end{lemma}
The proof of Lemma \ref{lem:ak} is in Appendix~\ref{app:lem-3}.

\noindent\textit{Step 2}: We next show that, conditioned on the $i^{\text{th}}$ trace agreeing with the current estimated prefix of length $k$, the correct symbol at position $k+1$ has a constant advantage over every incorrect symbol. That is, for $k\in\{0,1,\ldots,n-1\}$, define
$
U_{k+1}^{(i)} := \{R^{(i)} \leq n-k-1\},
$
which is the event that the $(k+1)$ symbol of the $i^{th}$ trace is uncorrupted. Following this, let 
$
\alpha_k
:=
\Pr\left(
U_{k+1}^{(i)} | A_k^{(i)}
\right).
$
Note that via the conditional independence of the traces, $\alpha_k$ does not depend on $i$. The following lemma then shows that $\alpha_k$ is bounded away from zero by a constant that depends only on $q$.
\begin{lemma}\label{lem:W1-temp}
Under channel $W_1$, for every $k\in\{0,1,\ldots,n-1\}$,
$
\alpha_k
\ge
\gamma_q,
\ \text{where}\
\gamma_q:=\frac{q-1}{2q-1}.
$
\end{lemma}
The proof of Lemma \ref{lem:W1-temp} is in Appendix \ref{app:lem-W1-temp}.

% \begin{proof}
% We prove this lemma by relating the event that position $k+1$ is uncorrupted, to the event $A_k$. Specifically, recall that, under channel $W_1$, position $k+1$ is uncorrupted if and only if
% $R\le n-k-1. 
% $ \arv{Is it? I thought you can also introduce the same symbol again post corruption \ldots The equality below may hence have to be replaced with an inequality (lower bound). An alternative is to \textit{define} the term ``uncorrupted" as being the setting where $R\leq n-k-1$.}
% On this event, the first $k$ coordinates are also uncorrupted, and hence the trace agrees with the true prefix of length $k$. Therefore,
% \[
% \alpha_k
% =
% \frac{
% P(R\le n-k-1)
% }{
% P(A_k)
% }.
% \]
% Further, since $R\sim\mathrm{Unif}([0:n])$, we have
% $
% P(R\le n-k-1)=\frac{n-k}{n+1}.
% $
% Thus, 
% $
% \alpha_k
% =
% \frac{n-k}{(n+1)P(A_k)}.
% $
% Next, using the upper bound on $P(A_k)$ from Lemma~4 \arv{I suggest putting references to lemmas/theorems inside a \textbackslash ref everywhere}, we obtain
% $
% \alpha_k
% \ge
% \frac{n-k}{
% n-k+1+\frac{1}{q-1}
% }.
% $
% Let $m=n-k$. Since $k\in\{0,1,\ldots,n-1\}$, we have $m\ge 1$. Hence,
% $
% \alpha_k
% \ge
% \frac{m}{m+1+\frac{1}{q-1}}.
% $
% Note that the right-hand side is increasing in $m$, and is therefore minimized at $m=1$. Consequently,
% $
% \alpha_k
% \ge
% \frac{1}{
% 2+\frac{1}{q-1}
% }
% =
% \frac{q-1}{2q-1},
% $
% completing the proof.
% \end{proof}
\noindent\textit{Step 3}: We now bound the probability of an incorrect mode decision at each step of PFM. Recall that at Step $k+1$, PFM considers only those traces whose first $k$ coordinates agree with the reconstructed prefix. Since we work with the all-zero input sequence, conditioned on the event that the first $k$ coordinates have been reconstructed correctly, this filtered set is precisely
\begin{align*}
S_k&:=\{i\in[N]: A_k^{(i)} \text{ occurs}\}\\
&=\{i\in[N]: y_1^{(i)}=0,\ldots,y_k^{(i)}=0\}.
\end{align*}
Further, for each $k\in\{0,1,\ldots,n-1\}$, let
$E_k$ be the event that PFM makes an incorrect decision at step $k+1$. Then, 
the following lemma bounds the probability of the event $E_k$.
\begin{lemma}
\label{lem:pfm-step-error}
For $a_q:=1-exp{\left(-\frac{(q-1)^2}{2(2q-1)^2}\right)}$, we have for every $k\in\{0,1,\ldots,n-1\}$ that
\[
P(E_k)
\le
(q-1)e^{\left(
-\frac{N(n-k+1)}{n+1}a_q
\right)}.
\]
\end{lemma}
The proof of Lemma \ref{lem:pfm-step-error} is in Appendix \ref{app:pfm-step-error}.
The proof of Theorem \ref{Thm:PFM} then follows by putting together the above lemmas.

\textit{Proof of Theorem~\ref{Thm:PFM}}:
 Note that the event that PFM fails is contained in the union of the events $E_k$ for $k\in\{0,1,\ldots,n-1\}$. Therefore, by Lemma~\ref{lem:pfm-step-error} and a union bound, we have
\begin{align*}
p_{e,\text{PFM}}^{(N)}(\mathbf{0})
&\le
P\left(\bigcup_{k=0}^{n-1}E_k\right)\\
&\le
\sum_{k=0}^{n-1}P(E_k) 
 \leq (q-1)\sum_{k=0}^{n-1}
e^{\left(
-\frac{N(n-k+1)}{n+1}a_q
\right)}.
\end{align*}
Next, choose
$
N=C(n+1),
$
for some constant $C > 0$ to be determined later. Then, we see that
\begin{align*}
    p_{e,\text{PFM}}^{(N)}(\mathbf{0}) &\leq (q-1)
\sum_{k=0}^{n-1} e^{\left(
-Ca_q(n-k+1)
\right)},\\
&\overset{(a)}{\leq} (q-1)
\sum_{j=2}^{n+1}
e^{-Ca_q j},\\
& \overset{(b)}{\leq} \frac{(q-1)e^{-Ca_q}}{1-e^{-Ca_q}},
\end{align*}
where $(a)$ follows via a change of variables and $(b)$ follows from the fact that $\sum_{j=2}^{n+1}e^{-Ca_q j}
\le
\sum_{j=1}^{\infty}e^{-Ca_q j}
=
\frac{e^{-Ca_q}}{1-e^{-Ca_q}}.$
Observe that if we choose
$
C
\ge
\frac{1}{a_q}
\log\left(
\frac{q-1+\delta}{\delta}
\right),
$ then the above expression is upper bounded by $\delta.$ In summary, by choosing $N \geq  \frac{1}{a_q}
\log\left(
\frac{q-1+\delta}{\delta}
\right) (n+1)$ traces for some $\delta > 0$, we obtain $ p_{e,\text{PFM}}^{(N)}(\mathbf{0}) \leq \delta$, completing the proof of the theorem. \qed

The above result completes the characterization of the trace complexity of channel $W_1.$ Combining Theorems~\ref{thm:lowerbound:w1} and~\ref{Thm:PFM}, we obtain the following corollary, which summarizes the main result of this section.
\begin{corollary}\label{cor:W1}
We have that for channel $W_1$, $T_{\delta}(n) = \Theta(n),$ and that  PFM achieves order-optimal trace complexity for this channel.
\end{corollary}

In the next subsection, we present experimental studies of the performance of some of the reconstruction algorithms under consideration for the channel $W_1$, for varying $n$ and $\delta$ values.
\subsection{Empirical Evaluation}
\label{sec:empirical-w1}

We apply \hyperref[alg:bwm]{BWM} and \hyperref[alg:pfm]{PFM} to numerically evaluate, via Monte Carlo
trials, the minimum number of traces $N$ required to reconstruct a
transmitted sequence with error probability at most $\delta$ over
channels $W_1$ and $W_2$. All experiments are conducted over the
binary alphabet ($q=2$). In each trial, a sequence
$\mathbf{X}\in\{0,1\}^n$ is generated uniformly at random and passed
independently through the chosen channel $N$ times to produce $N$
i.i.d.\ traces. The reconstruction algorithm produces an estimate
$\widehat{\mathbf{X}}$; the trial is a \emph{success} if
$\widehat{\mathbf{X}}=\mathbf{X}$ and a \emph{failure} otherwise.
For a fixed pair $(n, N)$, trials are drawn in adaptive batches of $100$, continuing until $100$ errors have accumulated or a cap of $100,000$ total trials is reached, and the empirical error probability $p'$ is computed from the accumulated trials and errors.

We construct $95\%$ confidence intervals using the Wilson score
interval \cite{wilson-score}: $\mathrm{CI}_{\mathrm{Wilson}}=[L,U],$ where
\begin{align*}\\
L&=\frac{p'+z_{\alpha/2}^2/2n_{\mathrm{tr}}-z_{\alpha/2}\sqrt{p'(1-p')/n_{\mathrm{tr}}+z_{\alpha/2}^2/4n^2_{\mathrm{tr}}}}
{1+z_{\alpha/2}^2/n_{\mathrm{tr}}},\\
U&=\frac{p'+z_{\alpha/2}^2/2n_{\mathrm{tr}}+z_{\alpha/2}\sqrt{p'(1-p')/n_{\mathrm{tr}}+z_{\alpha/2}^2/4n^2_{\mathrm{tr}}}}
{1+z_{\alpha/2}^2/n_{\mathrm{tr}}},
\end{align*}
where $n_{\mathrm{tr}}$ is the number of trials and
$z_{\alpha/2}=1.96$. For each parameter setting, we sweep $N$ and
identify three thresholds: $N_{\mathrm{central}}$, the smallest $N$
at which the point estimate $p'$ first drops below the target
$\delta$; $N_{\mathrm{lower}}$, the smallest $N$ at which the upper
end of the confidence interval drops below $\delta$ (a conservative
bound); and $N_{\mathrm{upper}}$, the smallest $N$ at which the lower
end of the confidence interval drops below $\delta$ (an optimistic
bound). In the figures that follow, $N_{\mathrm{central}}$ is the
solid curve and the shaded band spans
$[N_{\mathrm{lower}},N_{\mathrm{upper}}]$.\\

\noindent
\textbf{Trace complexity versus sequence length:}
We fix $\delta=0.01$ and vary $n$, recording the trace complexity $N_\delta(n)$ of BWM and PFM over $W_1$.

\begin{figure}[ht!]
    \centering
    \includegraphics[width=0.75\linewidth]{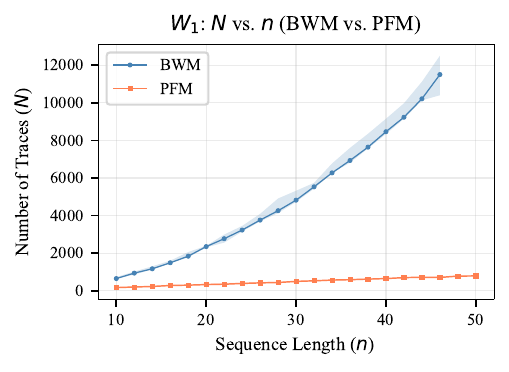}
    \caption{$N_\delta(n)$ over $W_1$: BWM versus PFM, with
    $\delta=0.01$.}
    \label{fig:nvsn-w1-algos}
\end{figure}
Figure~\ref{fig:nvsn-w1-algos} shows that PFM requires far fewer
traces than BWM over $W_1$, and that the gap widens with $n$. The plots are also visually 
consistent with the $O(n^2\log n)$ trace complexity guarantee for BWM in
Theorem~\ref{thm:bwm-w1} and the order-optimal $O(n)$ trace complexity guarantee for
PFM in Theorem~\ref{Thm:PFM}.\\ 

% \textbf{ Deepak: Should I keep the rest of the passage..?} The prefix-filtering step directly
% exploits the left-to-right structure of $W_1$, whereas BWM is
% bottlenecked by the least reliable coordinates near the end of the
% sequence.

\noindent
\textbf{Trace complexity versus target error probability:}
We now fix $n=20$ and vary $\delta\in[0.01,0.5]$, recording $N_\delta(n)$
for BWM and PFM over $W_1$.

\begin{figure}[H]
    \centering
    \includegraphics[width=0.75\linewidth]{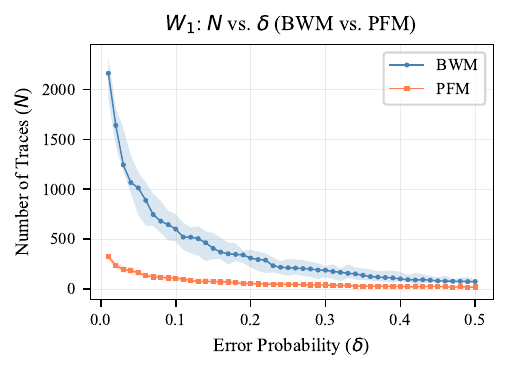}
    \caption{$N_\delta(n)$ over $W_1$: \hyperref[alg:bwm]{BWM} versus \hyperref[alg:pfm]{PFM}, with $n=20$}
    \label{fig:nvsdelta-w1-algos}
\end{figure}

As shown in Figure~\ref{fig:nvsdelta-w1-algos}, the required number
of traces decreases sharply as $\delta$ is relaxed from $0.01$
towards $0.1$ and then flattens for larger $\delta$. \hyperref[alg:pfm]{PFM} maintains a
substantial (and uniform) advantage over \hyperref[alg:bwm]{BWM} throughout the tested range.

\section{TrimAndExtend (Channel $W_2$)}
\label{sec:w2}
We now move to channel $W_2$. Our main contributions in this section are as follows: we first establish an $\Omega(n^2)$ lower bound on its trace complexity. We then prove that Bit-Wise Mode (BWM) algorithm (Algorithm~\ref{alg:bwm} introduced in the previous section) achieves a trace complexity of $O(n^2)$ for this channel. Together with the lower bound, this yields a tight $\Theta(n^2)$ characterization of the trace complexity, thereby establishing that BWM achieves order-optimal trace complexity for $W_2$.
% We first establish the lower bound.
\begin{theorem}\label{Thm:lowerbound:w2}
For channel $W_2$, and any fixed $\delta\in(0,1 - 1/q)$, we have
$
T_\delta(n)=\Omega(n^2).
$
\end{theorem}
    Similar to Theorem~\ref{thm:lowerbound:w1}, the proof of this theorem follows directly by evaluating the divergence term appearing in Lemma~\ref{lemma-lowerbound-channel-w1-w2}. We refer the reader to Appendix~\ref{proof:thm:lowerbound:w2} for the full proof. 
Next, we prove an upper bound on the trace complexity of $W_2$.
\begin{theorem}\label{thm:bwm:w2}
  For any constant $\delta \in (0,1)$, over channel $W_2$, the BWM
algorithm reconstructs the input sequence with probability at least
$1-\delta$ using $N = O(n^2)$ traces. Consequently,
$
T_{\delta,\mathrm{BWM}}(n) = O(n^2).
$
Furthermore, since
$T_\delta(n) \leq T_{\delta,\mathrm{BWM}}(n)$, channel $W_2$ has
trace complexity
$
T_\delta(n) = O(n^2).
$
\end{theorem}
The proof of Theorem \ref{thm:bwm:w2} is in Appendix \ref{app-bwm-w2}. The following corollary is then immediate from Theorems \ref{thm:bwm:w2} and \ref{Thm:lowerbound:w2}.
\begin{corollary}\label{cor:W2}
We have that for channel $W_2$, $T_{\delta}(n) = \Theta(n^2),$ and that BWM achieves order-optimal trace complexity for this channel.
\end{corollary}
\begin{remark}
  Note that unlike channel $W_1$, where the reliability of the coordinates decreases monotonically and the last coordinate is
left uncorrupted with probability only $\Theta(1/n)$, channel $W_2$ exhibits a fundamentally different behavior. In particular, the probability that coordinate $j$ remains uncorrupted is $ p_j = \frac{2j(n-j+1)}{(n+1)(n+2)},$ which is symmetric about the midpoint of the sequence. Consequently, the least reliable coordinates are the ``boundary''
coordinates, for which $p_j = \Theta(1/n)$. Since BWM estimates each coordinate independently, standard concentration arguments (which follow in Theorem~\ref{thm:bwm:w2} below) imply that recovering such coordinates reliably requires $\Theta(1/p_j^2) = \Theta(n^2)$
traces. This matches the $\Omega(n^2)$ lower bound established in Theorem~\ref{Thm:lowerbound:w2}. Therefore, unlike channel $W_1$, where the ``weakest'' coordinate is biased towards one end of the sequence
and the coordinate-wise estimation performed by BWM is
suboptimal, the coordinate-wise estimation performed by BWM is already sufficient to achieve the information-theoretic limit under channel $W_2$.
\end{remark}

In the next subsection, we present experimental studies of the performance of some of the reconstruction algorithms under consideration for the channel $W_2$, for varying $n$ and $\delta$ values.

\subsection{Empirical Evaluation}
\label{sec:empirical-w2}
We use the Monte Carlo procedure and confidence-interval thresholds
described in Section~\ref{sec:empirical-w1}, now for 
channel $W_2$.\\

\noindent
\textbf{Trace complexity versus sequence length.}
We fix $\delta=0.01$ and vary $n$, recording $N_\delta(n)$ for BWM
and PFM over $W_2$.

\begin{figure}[H]
    \centering
    \includegraphics[width=0.75\linewidth]{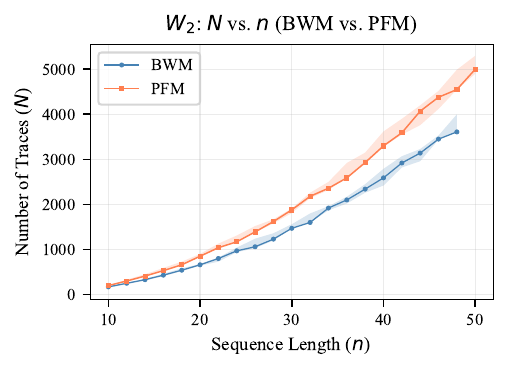}
    \caption{$N_\delta(n)$ over $W_2$: BWM versus PFM, with
    $\delta=0.01$.}
    \label{fig:nvsn-w2-algos}
\end{figure}

Figure~\ref{fig:nvsn-w2-algos} shows that BWM modestly outperforms
PFM over the tested range of sequence lengths. This agrees with the
role of BWM under $W_2$: its coordinate-wise estimates already attain
the order-optimal $O(n^2)$ trace complexity established in
Theorem~\ref{thm:bwm:w2}. In addition, the plots visually appear to confirm the analytical, asymptotic trace complexity of BWM. In contrast, filtering on a prefix does not
exploit the two-sided structure of $W_2$ as directly as it exploits
the one-sided structure of $W_1$. \\

\noindent
\textbf{Trace complexity versus target error probability.}
We fix $n=20$ and vary $\delta\in[0.01,0.5]$, recording $N_\delta(n)$
for BWM and PFM over $W_2$.

\begin{figure}[H]
    \centering
    \includegraphics[width=0.75\linewidth]{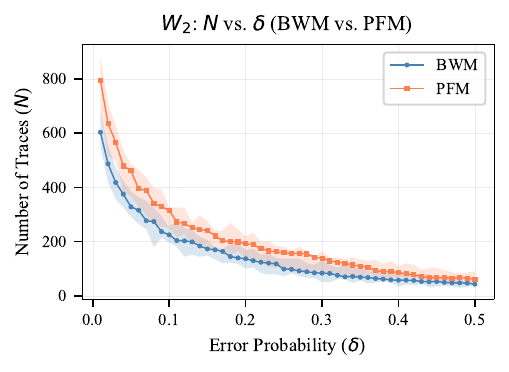}
    \caption{$N_\delta(n)$ over $W_2$: BWM versus PFM, with $n=20$.}
    \label{fig:nvsdelta-w2-algos}
\end{figure}

Figure~\ref{fig:nvsdelta-w2-algos} again shows a rapid decrease in
the required number of traces for small $\delta$, followed by a
flatter dependence as the target error probability is relaxed. BWM
requires fewer traces, in our experiments, than PFM throughout the displayed range.

\section{SuffixExtend$_{t}$(TrimSuffix) (Channel $W_3$)} \label{sec-w3}

We now turn to channel $W_3$. As before, we first present a simple
sequence reconstruction algorithm, TrimAndFindMode, for this channel that succeeds with
high probability, thus giving rise to an upper bound on the trace
complexity. We first recapitulate the channel law given in
\cite[Sec.~III-B]{bhardwaj}. Let $\ell(\u,\y)$ denote the length of
the longest common prefix of strings $\u\in\mathcal{X}^n$ and
$\y\in\mathcal{X}^\star$. The following lemma then holds.

\begin{lemma}[{\cite[Eq.~(3)]{bhardwaj}}]
    \label{lem:w3-calc}
    For any $\u\in\mathcal{X}^n$ and $\y\in\mathcal{X}^m$, for some
    $0\leq m\leq n+t$, we have
    \[
    W_3(\y|\u)
    =\frac{1}{(n+1)(t+1)}
    \sum_{k=(m-t)_+}^{\ell(\u,\y)}\frac{1}{q^{m-k}}.
    \]
\end{lemma}

We now present a simple sequence reconstruction algorithm,
{TrimAndFindMode}, over $W_3$; see
Algorithm~\ref{alg:trimmode}.

\begin{algorithm}[t]
\caption{TrimAndFindMode}
\label{alg:trimmode}
\begin{algorithmic}[1]
\Require Traces $\mathbf{y}^{(1)},\ldots,\mathbf{y}^{(N)}
\in\mathcal X^\star$
\Ensure Estimate $\widehat{\mathbf{x}}\in\mathcal X^n$

\State
$\mathcal S
\gets
\{i\in[N]:|\mathbf{y}^{(i)}|\geq n\}$

\For{$i\in\mathcal S$}
    \State
    $\overline{\mathbf{y}}^{(i)}
    \gets
    (y_1^{(i)},\ldots,y_n^{(i)})$
\EndFor

\State
$\widehat{\mathbf{x}}
\in
\arg\max_{\mathbf{z}\in\mathcal X^n}
\sum_{i\in\mathcal S}
\mathbf{1}\{\overline{\mathbf{y}}^{(i)}=\mathbf{z}\}$

\State
\Return $\widehat{\mathbf{x}}$

\end{algorithmic}
\end{algorithm}

The following lemma will be useful. Let
$\ell=\ell(\y):=\ell(\0,\y)$ when the sequence $\y$ is clear from
the context, and let $\lambda:=t+\ell-n+1$. Define
\[
A_\ell:=
\frac{q^{-\ell}}{(n+1)(t+1)(q-1)^2}
\left[q^{\lambda+1}-(\lambda+1)q+\lambda\right].
\]

\begin{lemma}
    \label{lem:trimmode-calc}
    For any $\widehat{\mathbf{y}}_1\in\mathcal{X}^n$, we have
    \[
    P(
    \widehat{\mathbf{y}}_1\mid\mathbf{0})
    =
    \begin{cases}
    A_\ell, & \text{if $\ell\geq n-t$},\\
    0, & \text{otherwise}.
    \end{cases}
    \]
\end{lemma}
The proof of Lemma \ref{lem:trimmode-calc} is in Appendix \ref{app-trimmode-calc}. We next show that
$P(\mathbf{0}\mid\mathbf{0})$ is strictly
larger than
$P(\widehat{\mathbf{y}}_1\mid\mathbf{0})$
for every $\widehat{\mathbf{y}}_1\neq\mathbf{0}$. By
Lemma~\ref{lem:trimmode-calc}, it suffices for the first statement to
show that this probability is increasing in $\ell$.

\begin{lemma}
    \label{lem:trimmode-gap}
    For any $\widehat{\mathbf{y}}_1\in\mathcal{X}^n$ such that
    $\widehat{\mathbf{y}}_1\neq\mathbf{0}$, we have
    \[
    P(\mathbf{0}\mid\mathbf{0})
    >
    P(
    \widehat{\mathbf{y}}_1\mid\mathbf{0}).
    \]
    Furthermore,
    \[
    P(\mathbf{0}\mid\mathbf{0})
    -\max_{\widehat{\mathbf{y}}_1\neq\mathbf{0}}
    P(
    \widehat{\mathbf{y}}_1\mid\mathbf{0})
    =\frac{q^{-n}}{n+1}.
    \]
\end{lemma}
The proof of Lemma \ref{lem:trimmode-gap} is in Appendix \ref{app-trimmode-gap}. Let $\overline{c}_q(n):=q^{-n}/(n+1)$. An argument analogous to the
earlier mode-decoder analyses yields the following theorem.

\begin{theorem}
    \label{thm:w3-ub}
    For any $\delta\in(0,1)$, the $\text{TrimAndFindMode}$ algorithm reconstructs the input sequence with probability at least $1-\delta$ using $O(n^3)$ traces.
    % \[
    % \Pr\!\left[
    % \mathrm{TrimAndFindMode}(\mathbf{y}_1,\ldots,\mathbf{y}_N)
    % \neq\mathbf{0}\mid\mathbf{0}
    % \right]
    % \leq\delta
    % \]
    % whenever
    % \[
    % N\geq
    % \frac{\ln(1/\delta)+n\ln q}{2(\overline{c}_q(n))^2}.
    % \]
    % In particular, $T_\delta(n)=O(n^3)$.
\end{theorem}

Given a channel $W$, we define
$$
d_n(W):=\min_{\mathbf{u}\neq \mathbf{u}'\in \mathcal{X}^n} d_{\text{TV}}\left(W(\cdot|\mathbf{u}), W(\cdot|\mathbf{u}')\right)
$$
to be the smallest total variational distance between channel transition probabilities corresponding to distinct channel input sequences. 

The following lemma, which holds via standard arguments analogous to the discussion after \cite[Thm. 1.2]{nazarov-peres}, then holds:
\begin{lemma}
	\label{lem:dist}
	For any channel $W$ and for all $\delta\in (0,1)$, we have that the trace complexity $T_\delta(n) = \Omega\left(\frac{1}{d_n(W)}\right)$.
\end{lemma} 
We finally state an asymptotic lower bound on $T_\delta(n)$. 
\begin{theorem}
    \label{thm:w3-lb}
    For any $\delta\in(0,1)$, $T_\delta(n)=\Omega(n)$.
\end{theorem}
The proof of Theorem \ref{thm:w3-lb} is in Appendix \ref{app-w3-lb}.

We observe from Theorem \ref{thm:w3-ub} and \ref{thm:w3-lb} that there exists a gap between the asymptotic bounds on trace complexity for $W_3$---we leave this question as open for future research. Another interesting direction for work is to incorporate noise in the form of mutations/substitutions, and rigorously derive the trace complexities of the resultant channels.

\bibliographystyle{IEEEtran}
\bibliography{references,submission-draft-2}
% \section*{Acknowledgments}
% The authors acknowledge the help of AI tools for calculations pertaining to portions of some proofs and for syntactical help with the generation of code for the numerical experiments.
% \arv{Please refine this if necessary.}

% Supplementary material: To improve readability, you must use a single-column format for the supplementary material.
\section*{Appendix}
\subsection{Proof of Lemma~\ref{lemma-lowerbound-channel-w1-w2}}\label{proof-lemma-lower-bound}

Fix a coordinate $j\in[n]$, and for each $p\in\mathcal{X}$, let
$\mathbf{x}^{(j,p)}$ denote the sequence that is equal to $p$ at
coordinate $j$ and equal to zero at all other coordinates. Let
$V\sim\mathrm{Unif}(\mathcal{X})$, and $\mathbf{x}^{(j,V)}$ be the transmitted sequence. Further, let
$\mathbf{y}_{1:N}=(\mathbf{y}^{(1)},\ldots,\mathbf{y}^{(N)})$
denote the collection of $N$ independent traces generated by the
channel. Then, conditioned on $V=p$, the distribution of $\mathbf{y}_{1:N}$ is therefore
\[
P_{n,j,p}^{(N)}=P_{n,j,p}^{\otimes N}.
\]

Now, suppose that $N$ traces are sufficient for the MAP decoder to
reconstruct the transmitted sequence with error probability at most
$\delta$. From Lemma~\ref{lem:worst-av}, the conditional error
probability of the MAP decoder is the same for every transmitted
sequence. In particular,
\[
p_{e,\mathrm{MAP}}^{(N)}
\left(\mathbf{x}^{(j,p)}\right)
\leq \delta,
\qquad
p\in\mathcal{X}.
\]
Hence, if $\widehat V$ denotes the $j$th coordinate of the sequence
returned by the MAP decoder, then
\[
P(\widehat V\neq V)\leq \delta.
\]

Next, from Fano's inequality \cite[Theorem~2.10.1]{cover1999elements}, together with the standard pairwise KL-divergence
bound for $q$-ary hypothesis testing and the symmetry of the channel with respect to the symbols in $\mathcal{X}$, we have
\[
\log q-h_b(\delta)-\delta\log(q-1)
\leq
N\max_{p\in\mathcal{X}\setminus\{0\}}
D\left(P_{n,j,0}\middle\|P_{n,j,p}\right),
\]
where we have also used the tensorization of KL divergence across
the $N$ independent traces. Following this, we have
\[
N\geq
\frac{
\log q-h_b(\delta)-\delta\log(q-1)
}{
\displaystyle
\max_{p\in\mathcal{X}\setminus\{0\}}
D\left(P_{n,j,0}\middle\|P_{n,j,p}\right)
}.
\]
Since this holds for every $j\in[n]$, maximizing over $j$ gives
\[
T_\delta(n)
\geq
\max_{j\in[n]}
\frac{
\log q-h_b(\delta)-\delta\log(q-1)
}{
\displaystyle
\max_{p\in\mathcal{X}\setminus\{0\}}
D\left(P_{n,j,0}\middle\|P_{n,j,p}\right)
},
\]
completing the proof.

\qed
\subsection{Proof of Theorem~\ref{thm:lowerbound:w1}} \label{app:thm-1}
We prove this theorem by applying Lemma~\ref{lemma-lowerbound-channel-w1-w2} and considering the coordinate $j=n$. Specifically, for any
$p\in\mathcal{X}\setminus\{0\}$, consider the two input sequences
\[
\mathbf{x}^{(n,0)}=(0,\ldots,0)
\quad\text{and}\quad
\mathbf{x}^{(n,p)}=(0,\ldots,0,p).
\]
Let $P_{n,n,0}$ and $P_{n,n,p}$ denote the corresponding
single-trace distributions under channel $W_1$. First, note that
\[
P_{n,n,0}\left(\mathbf{x}^{(n,0)}\right)
=
P_{n,n,p}\left(\mathbf{x}^{(n,p)}\right)
=
\frac{1}{n+1}\sum_{r=0}^{n}q^{-r}.
\]
Further,
\[
P_{n,n,0}\left(\mathbf{x}^{(n,p)}\right)
=
P_{n,n,p}\left(\mathbf{x}^{(n,0)}\right)
=
\frac{1}{n+1}\sum_{r=1}^{n}q^{-r}.
\]
For simplicity, let
\[
a:=\frac{1}{n+1}\sum_{r=0}^{n}q^{-r}
\quad\text{and}\quad
b:=\frac{1}{n+1}\sum_{r=1}^{n}q^{-r}.
\]
Note that for all other traces, the probabilities under $P_{n,n,0}$ and
$P_{n,n,p}$ are identical. Therefore,
\begin{align*}
D(P_{n,n,0}\|P_{n,n,p})
&=a\log\frac{a}{b}+b\log\frac{b}{a}\\
&=(a-b)\log\frac{a}{b}.
\end{align*}
Next, observe that
\[
a-b=\frac{1}{n+1},
\]
and
\[
\frac{a}{b}
=
\frac{\sum_{r=0}^{n}q^{-r}}
{\sum_{r=1}^{n}q^{-r}}
\leq q+1.
\]
Hence, we have
\[
D(P_{n,n,0}\|P_{n,n,p})
\leq
\frac{\log(q+1)}{n+1}.
\]
Additionally, since the above bound holds for every
$p\in\mathcal{X}\setminus\{0\}$, we have
\[
\max_{p\in\mathcal{X}\setminus\{0\}}
D(P_{n,n,0}\|P_{n,n,p})
\leq
\frac{\log(q+1)}{n+1}.
\]
Finally, applying Lemma~\ref{lemma-lowerbound-channel-w1-w2} with $j=n$, we obtain
\[
T_\delta(n)
\geq
\frac{\log q-h_b(\delta)-\delta\log(q-1)}
{\log(q+1)}
(n+1).
\]
Note that, for any fixed $\delta\in(0,1-1/q)$, the multiplicative factor
above is a positive constant. Hence,
\[
T_\delta(n)=\Omega(n),
\]
completing the proof. \qed
\subsection{Proof of Theorem~\ref{thm:bwm-w1}}\label{proof:thm-bwm-w1}
First, fix a coordinate $j\in[n]$ and an incorrect symbol $a\neq x_j$. Note that, under channel $W_1$, the symbol $x_j$ is preserved whenever $R\leq n-j$, which occurs with probability $(n-j+1)/(n+1)$. If $R>n-j$, then the symbol at coordinate $j$ is replaced by a symbol drawn uniformly from $\mathcal X$. Therefore, for any trace $i \in [N]$,
\[
P(y_j^{(i)}=x_j)
=
\frac{n-j+1}{n+1}
+
\frac{j}{q(n+1)},
\]
whereas, for any $a\neq x_j$,
\[
P(y_j^{(i)}=a)
=
\frac{j}{q(n+1)}.
\]
Hence,
\[
P(y_j^{(i)}=x_j)-P(y_j^{(i)}=a)
=
\frac{n-j+1}{n+1}
\geq
\frac{1}{n+1}.
\]
Next, for each trace $i\in[N]$, define
\[
Z_i^{(j,a)}
:=
\mathbf 1\{y_j^{(i)}=x_j\}
-
\mathbf 1\{y_j^{(i)}=a\}.
\]
Then, we see $Z_i^{(j,a)}\in[-1,1]$, and
\[
\mathbb E[Z_i^{(j,a)}]
=
P(y_j^{(i)}=x_j)-P(y_j^{(i)}=a)
\geq
\frac{1}{n+1}.
\]
 Further, the event that the incorrect symbol $a$ appears at least as frequently as the true symbol $x_j$ at coordinate $j$ precisely corresponds to the event
$
\left\{\sum_{i=1}^{N} Z_i^{(j,a)} \leq 0\right\}.
$

Therefore, by Hoeffding's inequality \cite[Proposition~2.5]{wainwright2019high}, we have
\[
P\left(
\sum_{i=1}^{N} Z_i^{(j,a)} \leq 0
\right)
\leq
\exp\left(
-\frac{N}{2(n+1)^2}
\right).
\]
Taking a union bound over all $a\in\mathcal X\setminus\{x_j\}$ gives
\[
P(\widehat{x}_j\neq x_j)
\leq
(q-1)
\exp\left(
-\frac{N}{2(n+1)^2}
\right).
\]
Finally, taking a union bound over all coordinates $j\in[n]$, we obtain
\[
P(\widehat{\mathbf{x}}\neq \mathbf{x})
\leq
n(q-1)
\exp\left(
-\frac{N}{2(n+1)^2}
\right).
\]
Thus, we can conclude that if
\[
N
\geq
2(n+1)^2
\log\left(\frac{n(q-1)}{\delta}\right),
\]
we must have
$
P(\widehat{\mathbf{x}}\neq \mathbf{x})
\leq \delta,$
completing the proof of the theorem.
\qed
\subsection{Proof of Lemma \ref{lem:ak}}\label{app:lem-3}
Consider any trace $y^{(i)}$ and let $R$ be the random variable denoting the number of trimmed symbols in this trace. If $R\le n-k$, then none of the first $k$ symbols are corrupted, and the event $A_k^{(i)}$ occurs with probability $1$. On the other hand, if $R>n-k$, then exactly $R-(n-k)$ symbols among the first $k$ positions are replaced by independent uniform symbols from $\mathcal X$. In this case, the event $A_k^{(i)}$ occurs only if all these replaced symbols are equal to $0$ (since $\mathbf{x} = \mathbf{0}$), which happens with probability
$
q^{-(R-(n-k))}.
$
Therefore, since $R\sim \mathrm{Unif}([0:n])$, we have 
\begin{align*}
P(A_k)
&=
\frac{1}{n+1}
\left(
\sum_{r=0}^{n-k}1
+
\sum_{r=n-k+1}^{n}
q^{-(r-(n-k))}
\right)\\
&=
\frac{1}{n+1}
\left(
n-k+1
+
\sum_{t=1}^{k}q^{-t}
\right).
\end{align*}
Consequently, the lower bound follows directly by retaining only the first term:
\[
P(A_k)
\ge
\frac{n-k+1}{n+1}.
\]
Next, for the upper bound, we use the fact that
\[
\sum_{t=1}^{k}q^{-t}
\le
\sum_{t=1}^{\infty}q^{-t}
=
\frac{1}{q-1}.
\]
Thus,
\[
P(A_k)
\le
\frac{1}{n+1}
\left(
n-k+1+\frac{1}{q-1}
\right),
\]
completing the proof. \qed
\section{Proof of Theorem~\ref{Thm:lowerbound:w2}}\label{proof:thm:lowerbound:w2}
Consider the coordinate $j=n$, and for any
$p\in\mathcal{X}\setminus\{0\}$, consider the two input sequences
\[
\mathbf{x}^{(n,0)}=(0,\ldots,0)
\qquad\text{and}\qquad
\mathbf{x}^{(n,p)}=(0,\ldots,0,p).
\]
Let $P_{n,n,0}$ and $P_{n,n,p}$ denote the corresponding
single-trace distributions under channel $W_2$. Recall that $(R_1,R_2)$ is sampled uniformly from the set
\[
\mathcal{R}_n
:=
\{(r_1,r_2)\in\mathbb{Z}_{\geq 0}^2:r_1+r_2\leq n\},
\]
where
\[
|\mathcal{R}_n|
=
\frac{(n+1)(n+2)}{2}.
\]
Next, note that for the two input sequences considered above, the channel output
distributions differ only when the last input coordinate is left
uncorrupted. This occurs when $R_2=0$ and $R_1\leq n-1$.
Consequently, the probability of this event is
\[
\frac{n}{|\mathcal{R}_n|}
=
\frac{2n}{(n+1)(n+2)}.
\]
A direct evaluation of the KL divergence then gives
\[
D(P_{n,n,0}\|P_{n,n,p})
\leq
\frac{2}{n+2}
\log\left(
1+
\frac{2nq}
{(n+1)(n+2)-2n}
\right).
\]
Additionally, observe that
\[
\frac{2nq}
{(n+1)(n+2)-2n}
=
\frac{2nq}{n^2+n+2}
=
O\left(\frac{1}{n}\right).
\]
Using $\log(1+x)\leq x/\ln 2$ for $x\geq0$, we obtain
\begin{align*}
D(P_{n,n,0}\|P_{n,n,p})
&\leq
\frac{2}{n+2}
\frac{1}{\ln 2}
\frac{2nq}{n^2+n+2} \\
&=
O\left(\frac{1}{n^2}\right).
\end{align*}
Since the above bound is independent of
$p\in\mathcal{X}\setminus\{0\}$, it follows that
\[
\max_{p\in\mathcal{X}\setminus\{0\}}
D(P_{n,n,0}\|P_{n,n,p})
=
O\left(\frac{1}{n^2}\right).
\]
Finally, applying Lemma~\ref{lemma-lowerbound-channel-w1-w2} with $j=n$, we obtain
\[
T_\delta(n)
\geq
\frac{
\log q-h_b(\delta)-\delta\log(q-1)
}{
\displaystyle
\max_{p\in\mathcal{X}\setminus\{0\}}
D(P_{n,n,0}\|P_{n,n,p})
}.
\]
Notice that for any fixed $\delta\in(0,1-1/q)$, the numerator is a
positive constant. Hence,
\[
T_\delta(n)=\Omega(n^2),
\]
completing the proof.
\qed
\subsection{Proof of Lemma \ref{lem:W1-temp}}\label{app:lem-W1-temp}
We prove this lemma by relating the event that position $k+1$ is uncorrupted to the event $A_k^{(i)}$. Here, we say that position $k+1$ is uncorrupted if it is unaffected by the trimming-and-extension operation, which occurs when
$
R^{(i)} \leq n-k-1.
$
On this event, the first $k$ coordinates are also uncorrupted, and hence the $i$th trace agrees with the true prefix of length $k$. Therefore,
\[
\{R^{(i)} \leq n-k-1\} \subseteq A_k^{(i)}.
\]
It then follows that
\begin{align*}
\alpha_k
&= \Pr\!\left(R^{(i)} \leq n-k-1 \,\middle|\, A_k^{(i)}\right) \\
&= \frac{\Pr\!\left(R^{(i)} \leq n-k-1\right)}
{\Pr\!\left(A_k^{(i)}\right)} \\
&= \frac{\Pr\!\left(R^{(i)} \leq n-k-1\right)}
{P(A_k)},
\end{align*}
where the last equality follows because the traces are identically distributed. Next, since $R^{(i)} \sim \operatorname{Unif}([0:n])$, we have
\[
\Pr\!\left(R^{(i)} \leq n-k-1\right)
=
\frac{n-k}{n+1}.
\]
Thus,
\[
\alpha_k
=
\frac{n-k}{(n+1)P(A_k)}.
\]
Further, using the upper bound on $P(A_k)$ from Lemma~\ref{lem:ak}, we obtain
\[
\alpha_k
\geq
\frac{n-k}
{n-k+1+\frac{1}{q-1}}.
\]
Furthermore, let $m=n-k$. Since $k\in\{0,1,\ldots,n-1\}$, we have $m\geq 1$. Hence,
\[
\alpha_k
\geq
\frac{m}
{m+1+\frac{1}{q-1}}.
\]
Note that the right-hand side is increasing in $m$ and is therefore minimized at $m=1$. Consequently,
\[
\alpha_k
\geq
\frac{1}
{2+\frac{1}{q-1}}
=
\frac{q-1}{2q-1}
=
\gamma_q,
\]
which completes the proof. \qed
\subsection{Proof of Lemma \ref{lem:pfm-step-error}}\label{app:pfm-step-error}
Consider $k\in\{0,1,\ldots,n-1\}$ and an incorrect symbol
$p\in\mathcal{X}\setminus\{0\}$. We first bound the probability
that $p$ appears at least as frequently as the correct symbol
$0$ among the traces in $S_k$.

Consider an arbitrary trace
$i\in S_k$. Conditioned on $A_k^{(i)}$, the probability that the
symbol at position $k+1$ is $0$ is
$
\alpha_k+\frac{1-\alpha_k}{q},
$
whereas the probability that it is equal to $p$ is
$
\frac{1-\alpha_k}{q}.
$
Hence,
\[
P\!\left(y_{k+1}^{(i)}=0\,\middle|\,A_k^{(i)}\right)
-
\Pr\!\left(y_{k+1}^{(i)}=p\,\middle|\,A_k^{(i)}\right)
=
\alpha_k,
\]
and by Lemma~\ref{lem:W1-temp}, we have
$
\alpha_k\ge\gamma_q.
$ Next, for each $i\in S_k$ and $p\in\mathcal{X}\setminus\{0\}$,
define
\[
Z_{i,k}^{(p)}
:=
\mathbf{1}\{y_{k+1}^{(i)}=0\}
-
\mathbf{1}\{y_{k+1}^{(i)}=p\}.
\]
Note that 
$
Z_{i,k}^{(p)}\in[-1,1],
$
and
$
\mathbb{E}\!\left[
Z_{i,k}^{(p)}
\,\middle|\,
i\in S_k
\right]
\ge
\gamma_q.
$
Further, the event that $p$ appears at least as frequently as
$0$ among the traces in $S_k$ is precisely
\[
\left\{
\sum_{i\in S_k}
Z_{i,k}^{(p)}
\le0
\right\}.
\]
In addition, conditioning on $|S_k|=s$, the random variables
$\{Z_{i,k}^{(p)}:i\in S_k\}$ are independent and take values
in $[-1,1]$. Hence, by Hoeffding's inequality\cite[Proposition~2.5]{wainwright2019high},
\[
P\!\left(
\sum_{i\in S_k}
Z_{i,k}^{(p)}
\le0
\,\middle|\,
|S_k|=s
\right)
\le
\exp\!\left(
-\frac{s\gamma_q^2}{2}
\right).
\]
Next, defining
$
\beta_q:=\frac{\gamma_q^2}{2},
$
we obtain
\[
\begin{aligned}
&P\left(
\sum_{i\in S_k}
Z_{i,k}^{(p)}
\le0
\right)\\
\hspace{0.5in}  &=
\sum_{s\ge0}
P\left(
\sum_{i\in S_k}
Z_{i,k}^{(p)}
\le0
\,\middle|\,
|S_k|=s
\right)
P(|S_k|=s)
\\
\hspace{0.5in} &\le
\sum_{s\ge0}
e^{-s\beta_q}P(|S_k|=s)
\\
\hspace{0.5in} &=
\mathbb{E}\!\left[e^{-\beta_q|S_k|}\right].
\end{aligned}
\]
Furthermore, since the traces are independent, we have
\[
|S_k|
\sim
\operatorname{Binomial}(N,P(A_k)).
\]
Therefore,
\begin{align*}
\mathbb{E}\!\left[e^{-\beta_q|S_k|}\right]
&=
\left(
1-P(A_k)+P(A_k)e^{-\beta_q}
\right)^N \\
&=
\left(
1-P(A_k)(1-e^{-\beta_q})
\right)^N.
\end{align*}
Next, let
$
a_q:=1-e^{-\beta_q}>0.
$
Using the standard inequality
$$
(1-u)^N\le e^{-Nu},
$$
we obtain
\[
\mathbb{E}\!\left[e^{-\beta_q|S_k|}\right]
\le
\exp\!\left(
-NP(A_k)a_q
\right).
\]
Recall from Lemma~\ref{lem:ak} that
\[
P(A_k)
\ge
\frac{n-k+1}{n+1}.
\]
Substituting this into the above and taking a union bound over
all $p\in\mathcal{X}\setminus\{0\}$ yields
\[
P(E_k)
\le
(q-1)
\exp\!\left(
-
N
\frac{n-k+1}{n+1}
a_q
\right),
\]
completing the proof. \qed

\subsection{Proof of Theorem \ref{thm:bwm:w2}}
\label{app-bwm-w2}
    Recall from Lemma~\ref{lem:worst-av}, it suffices to analyze the all-zero input sequence $\mathbf{x}=\mathbf{0}$. First, fix a coordinate $j\in[n]$. Now observe that under channel $W_2$, the coordinate $j$ is left uncorrupted if and only if
$
R_1\le j-1$
\text{and} $
R_2\le n-j.
$
Since $(R_1,R_2)$ is sampled uniformly from the set of all nonnegative integer pairs satisfying $R_1+R_2\le n$, we have
\[
p_j:=P(R_1\le j-1\text{ and }R_2\le n-j)
=
\frac{2j(n-j+1)}{(n+1)(n+2)}.
\]
Then, for any trace $i \in \{1,2,\ldots,N\}$ we have 
$
P(y_j^{(i)}=0)
=
p_j+\frac{1-p_j}{q},
$
whereas, for every other symbol $a\in\mathcal X\setminus\{0\}$, we have
$
P(y_j^{(i)}=a)
=
\frac{1-p_j}{q}.
$
Hence, the difference
$
P(y_j^{(i)}=0)-P(y_j^{(i)}=a)$ equals 
$p_j$.
 Next, for each trace $i\in[N]$ and each $a\neq 0$, define
$
Z_i^{(j,a)}
:=
\mathbf 1\{y_j^{(i)}=0\}
-
\mathbf 1\{y_j^{(i)}=a\}.
$
We see that $Z_i^{(j,a)}\in[-1,1]$ and
$
\mathbb E[Z_i^{(j,a)}]=p_j.
$
Further, if BWM incorrectly estimates coordinate $j$ as $a$, then
$
\sum_{i=1}^{N} Z_i^{(j,a)}\le 0.
$
Therefore, by Hoeffding's inequality \cite[Proposition~2.5]{wainwright2019high},
\[
P\left(
\sum_{i=1}^{N}Z_i^{(j,a)}\le 0
\right)
\le
e^{
-{Np_j^2}/{2}
}.
\]
Next, taking a union bound over all $a\in\mathcal X\setminus\{0\}$ gives
\[
P(\widehat{x}_j\neq 0)
\le
(q-1)e^{
-{Np_j^2}/{2}}.
\]
Furthermore, by applying a union bound over all coordinates, we have 
\begin{align*}
P(\widehat{\mathbf{x}}\neq \mathbf{0})
&\le
(q-1)\sum_{j=1}^{n}
\exp{
-\frac{Np_j^2}{2}
}, \\
&\overset{(a)}{\leq} 2(q-1)
\sum_{j=1}^{\lceil n/2\rceil}
\exp{
-\frac{Np_j^2}{2}
}, \\
&\overset{(b)}{\leq} 2(q-1)
\sum_{j=1}^{\infty}
\exp{
-\frac{Njn^2}{2(n+1)^2(n+2)^2}
},\\
&\overset{(c)}{\leq} 2(q-1)\frac{\exp{
-\frac{Njn^2}{2(n+1)^2(n+2)^2}
}}{1 - \exp{
-\frac{Njn^2}{2(n+1)^2(n+2)^2}
}},
\end{align*}
where $(a)$ follows from the symmetry $p_j = p_{n-j+1}$, $(b)$ follows from the fact that $p_j \geq \frac{jn}{(n+1)(n+2)},$ for  $j\leq \lceil n/2\rceil$, which in turn implies that $p_j^2 \geq \frac{j^2 n^2}{(n+1)^2 (n+2)^2} \geq \frac{jn^2}{(n+1)^2 (n+2)^2,}$ and finally $(c)$ follows from standard calculations.
Note that the above probability is less than $\delta$ if
\[
N
\ge
\frac{2(n+1)^2(n+2)^2}{n^2}
\log\left(
\frac{2(q-1)+\delta}{\delta}
\right),
\]
or, in other words, the trace complexity of channel $W_2$ under BWM obeys
$
T_{\delta, BWM}(n)=O(n^2).
$
Furthermore, since
$
T_\delta(n)\leq T_{\delta,\mathrm{BWM}}(n),
$
we also obtain $T_\delta(n)=O(n^2)$,
thereby completing the proof.
\qed
\subsection{Proof of Lemma \ref{lem:trimmode-calc}}
\label{app-trimmode-calc}
Suppose that
$P(\widehat{\mathbf{y}}_1\mid\mathbf{0})>0$.
If $n-\ell(\widehat{\mathbf{y}}_1)\geq t+1$, then the length of
$\overline{\mathbf{y}}_1$ that was trimmed to obtain
$\widehat{\mathbf{y}}_1$ is at most $n-1$, leading to a contradiction.
When $\ell(\widehat{\mathbf{y}}_1)\geq n-t$, we have
\begin{align*}
P(\widehat{\mathbf{y}}_1\mid\mathbf{0})
&=\sum_{r=n-\ell}^{t}P_R(r)P(E\geq r)q^{r-n}\\
&=\frac{1}{(n+1)(t+1)}
  \sum_{r=n-\ell}^{t}(t-r+1)q^{r-n}\\
&=\frac{q^{-n+t+1}}{(n+1)(t+1)}
  \sum_{r=1}^{t+\ell-n+1}rq^{-r}=A_\ell,
\end{align*}
which proves the lemma. \qed
\section{Proof of Lemma \ref{lem:trimmode-gap}}
\label{app-trimmode-gap}
For the first statement, following Lemma~\ref{lem:trimmode-calc}, it
suffices to show that
\[
\begin{aligned}
g(\ell)
&:=q^{-\ell}\bigl(t+\ell-n+1-(t+\ell-n+2)q\bigr)\\
&\quad-q^{-\ell+1}
\bigl(t+\ell-n-(t+\ell-n+1)q\bigr)
\end{aligned}
\]
satisfies $g(\ell)>0$ for every $n-t\leq\ell\leq n$. Indeed,
\[
g(\ell)=q^{-\ell}(t+\ell-n+1)(q-1)^2>0
\]
for all $q>1$ and $n-t\leq\ell\leq n$. The second statement follows
because the maximum over $\widehat{\mathbf{y}}_1\neq\mathbf{0}$ is
attained by any $\widehat{\mathbf{y}}_1$ with
$\ell(\widehat{\mathbf{y}}_1)=n-1$. \qed
\subsection{Proof of Theorem \ref{thm:w3-lb}}
\label{app-w3-lb}
Define
\[
d_n(W_3):=
\min_{\mathbf{u}\neq\mathbf{0}}
d_{\mathrm{TV}}\!\left(
W_3(\cdot\mid\mathbf{0}),W_3(\cdot\mid\mathbf{u})
\right).
\]
By Lemma~\ref{lem:w3-calc}, for any $\mathbf{u}\neq\mathbf{0}$,
define
\[
S_{\mathbf{u}}(\mathbf{y})
:=\sum_{k=(|\mathbf{y}|-t)_+}^{\ell(\mathbf{u},\mathbf{y})}
q^{k-|\mathbf{y}|}.
\]
Then
\begin{align*}
&d_{\mathrm{TV}}\!\left(
W_3(\cdot\mid\mathbf{0}),W_3(\cdot\mid\mathbf{u})
\right)\\
&\quad=\frac{1}{2(n+1)(t+1)}
\sum_{\mathbf{y}\in\mathcal{X}^{\star}}
\left|S_{\mathbf{0}}(\mathbf{y})-S_{\mathbf{u}}(\mathbf{y})\right|,
\end{align*}
Intuitively, each summand is small when
$\ell(\mathbf{0},\mathbf{y})$ is close to
$\ell(\mathbf{u},\mathbf{y})$. In particular, choosing
$\mathbf{u}=0^{n-1}1$ gives, after algebraic simplification,
\[
d_n(W_3)\leq\frac{1}{2(n+1)(t+1)}.
\]
An application of Lemma \ref{lem:dist} then proves the theorem. \qed

\end{document}